\documentclass[11pt,a4paper]{article}
\usepackage{amssymb,amsmath,amsthm}
\usepackage{graphicx,color}
\usepackage{xspace}

\newtheorem{theorem}{Theorem}

\newtheorem{corollary}{Corollary}
\newtheorem{lemma}{Lemma}

\DeclareMathOperator{\operatorClassP}{P}
\newcommand{\classP}{\ensuremath{\operatorClassP}}
\DeclareMathOperator{\operatorClassXP}{XP}
\newcommand{\classXP}{\ensuremath{\operatorClassXP}}
\DeclareMathOperator{\operatorClassNP}{NP}
\newcommand{\classNP}{\ensuremath{\operatorClassNP}}

\DeclareMathOperator{\operatorClassFPT}{FPT}
\newcommand{\classFPT}{\ensuremath{\operatorClassFPT}}
\DeclareMathOperator{\operatorClassW}{W}
\newcommand{\classW}[1]{\ensuremath{\operatorClassW[#1]}}

\newcommand{\DAKC}{\textsc{Dir-AKC}\xspace}

\newcommand{\psc}{\textsc{Partial Set Cover}\xspace}

\begin{document}

\title{Parameterized Complexity of the Anchored $k$-Core Problem for Directed Graphs\footnote{This work is supported by the European Research Council (ERC) via grant Rigorous Theory of Preprocessing, reference
267959 and by NSF CAREER award 1053605, NSF grant CCF-1161626, ONR YIP award N000141110662, DARPA/AFOSR grant FA9550-12-1-0423, a University
of Maryland Research and Scholarship Award (RASA) and a Summer International Research Fellowship from the University of
Maryland.
}}

\author{%
Rajesh Chitnis\thanks{Department of Computer Science , University of Maryland at College Park, USA. Email: {\tt rchitnis@cs.umd.edu.}}
\and
Fedor V. Fomin\thanks{Department of Informatics, University of Bergen, PB 7803, 5020 Bergen, Norway. Email: {\tt{\{fedor.fomin, petr.golovach\}@ii.uib.no}}}
\addtocounter{footnote}{-1}
\and 
Petr A. Golovach\footnotemark
}

\date{}

\maketitle

\begin{abstract} Motivated by the study of  unraveling processes in social networks,
Bhawalkar,   Kleinberg,   Lewi,   Roughgarden,  and Sharma [ICALP 2012] introduced the \textsc{Anchored $k$-Core} problem,
where the task is  for a given graph $G$ and integers $b, k$, and $p$ to find an induced subgraph $H$  with at least $p$
vertices (the core) such that all but at most $b$ vertices (called anchors) of  $H$ are of
degree at least $k$. In this paper, we extend the notion of $k$-core to directed graphs and provide a number of new
algorithmic and complexity results for the directed version of the problem. We show that
\begin{itemize}
\item
The decision version of the problem is   \classNP-complete for every $k\geq 1$  even if the input graph is restricted to be a planar directed acyclic graph of maximum degree at
most $k+2$.
\item The problem is fixed parameter tractable (\classFPT) parameterized by the size of the core $p$ for  $k=1$, and \classW1-hard for $k\geq 2$.
\item When the maximum degree of the graph is at most $\Delta$,   the
 problem is  \classFPT\  parameterized by $p+\Delta$ if $k\geq \frac{\Delta}{2}$.
\end{itemize}
\end{abstract}

\section{Introduction}\label{sec:intro}
The anchored $k$-core problem can be explained by the following illustrative example. We want
to organize a workshop on Theory of Social Networks. We send invitations to most  distinguished researchers in the area and
received many  replies of the following nature: ``Yes, in theory, I would be happy to come but my final decision  depends on
how many   people I know will be there."  Thus  we have a list of  tentative participants, but some of them can cancel their
participation and we are afraid that the cancellation process
may escalate. On the other hand, we also have  limited funds to reimburse travel expenses for
a small number of participants, which we believe, will guarantee their participation. Thus we want to ``anchor" a small subset
of  participants whose guaranteed participation would prevent the unraveling process, and by
fixing a small group   we hope to minimize the number of cancellations, or equivalently,
maximize the number of participants, or the core.

Unraveling processes are common for social networks where the behavior of an individual is often influenced by the actions of
her/his friends. New events occur quite often in social networks: some examples are usage of a particular cell phone brand,
adoption of a new drug within the medical profession, or the rise of a political movement in an unstable society. To estimate
whether these events or ideas spread extensively or die out soon, one has  to model and study the dynamics of \emph{influence
propagation} in social networks. Social networks are generally represented by making use of  undirected or directed graphs,
where the edge set represents  the relationship between individuals in the network. Undirected graph  model works fine for
some networks, say Facebook, but the nature of interaction on some social networks such as Twitter is asymmetrical:  the fact
that  user $A$ follows  user $B$ does not imply that that user $B$ also follows $A$.\footnote{The first author follows LeBron
James on Twitter (and so do 8,017,911 other people), but he only follows 302 people with the first author not being one of
them.}  In this case, it is more appropriate to model interactions in  the network by \textbf{directed} graphs. We add a
directed edge $(u,v)$ if $v$ follows $u$.

In this work we are interested in the  model of \emph{user engagement}, where each individual with less than $k$ people to
follow (or equivalently whose in-degree is less than $k$) drops out of the network. This
process can be contagious, and may affect even those individuals who initially were linked to more than $k$ people, say follow
on Twitter. An extreme example of this was given by Schelling (see page 17 of ~\cite{schelling2006micromotives}): consider a
directed path on $n$ vertices and let $k=1$. The left-endpoint has in-degree zero, it drops out and now the in-degree of its
only out-neighbor in the path becomes zero and it drops out as well. It is not hard to see that this way the whole network
eventually drops out as the result of a \emph{cascade of iterated withdrawals}. In general at the end of all the iterated
withdrawals the remaining engaged individuals form a unique maximal induced subgraph whose minimum in-degree is at least $k$.
This is called as the \emph{$k$-core} and is a well-known concept in the theory of social networks. It was introduced by
Seidman~\cite{seidman-k-core} and also been studied in various social sciences
literature~\cite{chwe1999structure,chwe2000communication}.

\medskip
\noindent\textbf{Preventing Unraveling:} The unraveling process described above in Schelling's example of a directed path can
be  highly undesirable in many scenarios.  How can one attempt to prevent this unraveling? In Schelling's example it is easy
to see: if we ``buy" the left end-point person into being engaged then the whole path becomes engaged. In general we overcome
the issue of unraveling by allowing some ``anchors":  these are the vertices that remain engaged irrespective of their
payoffs. This can be achieved by giving them extra incentives or discounts. The hope is that with a few \emph{anchors} we can
now ensure a large subgraph remains engaged. This subgraph is  called as the \emph{anchored $k$-core}: each non-anchor vertex
in this induced subgraph must have in-degree at least $k$ while the anchored vertices can have arbitrary in-degrees. The
problem of identifying $k$-cores in a network also has the following  game-theoretical interpretation   introduced by
Bhawalkar et al.~\cite{BhawalkarKLRS12}: each user in the social network pays a cost of $k$ to remain engaged. On the other
hand, he/she receives a profit of one from every neighbor who is engaged. The ``network effects" come into play, and an
individual decides to remain engaged if has non-negative payoff, i.e., it has at least $k$ in-neighbors who are engaged. The
$k$-core can be viewed as the unique maximal equilibrium in this model.

Bhawalkar et al.~\cite{BhawalkarKLRS12} introduced  the \textsc{Anchored $k$-Core} problem for (undirected) graphs. In the
\textsc{Anchored $k$-Core} problem the input is an undirected graph $G=(V,E)$ and integers $b, k$, and the task is to find  an
induced subgraph $H$ of maximum size with all vertices but at most $b$ (which are anchored)
to be of degree at least $k$. In this work we extend the
notion of {anchored $k$-core} to directed graphs.  We are interested in the case, when  in-degrees of all but $b$ vertices of
$H$ are at least  $k$.  More formally, we study the following parameterized version of the problem.
\begin{center}
\noindent\framebox{\begin{minipage}{4.50in} { \textsc{Directed Anchored $k$-Core}  (\DAKC)}\\
\emph{Input}: A directed graph $G=(V,E)$ and integers $b, k,p$. \\
\emph{Parameter~1}: $b$.\\
\emph{Parameter~2}: $k$.\\
\emph{Parameter~3}: $p$.\\
\emph{Question}:  Do there exist sets of vertices $A\subseteq H\subseteq V(G)$ such that $|A|\leq b$, $|H|\geq p$, and every
$v\in H\setminus a$ satisfies $d^{-}_{G[H]}(v)\geq k$?
\end{minipage}}
\end{center}
We will call the set $A$ as the \emph{anchors}, the graph $H$ as the \emph{anchored $k$-core}. Note that the undirected
version of \textsc{Anchored $k$-Core} problem can be modeled by the directed version: simply replace each edge $\{u,v\}$ by
arcs $(u,v)$ and $(v,u)$. Keeping the parameters $b,k,p$ unchanged it is now easy to see that the two instances are
equivalent.

\medskip
\noindent\textbf{Parameterized Complexity:} We are mainly interested in the parameterized complexity of  \textsc{Anchored
$k$-Core}. For the general background, we refer to the books by Downey and Fellows~\cite{downey-fellows-book},  Flum and
Grohe~\cite{flum-grohe-book} and Niedermeier~\cite{niedermeier-book}. Parameterized complexity is basically a two dimensional
framework for studying the computational complexity of a problem. One dimension is the input size $n$ and another one is a
parameter $k$. A problem is said to be \emph{fixed parameter tractable} (or \classFPT) if it can be solved in time $f(k)\cdot
n^{O(1)}$ for some function $f$. A problem is said to be in \classXP, if it can be solved in time $O(n^{f(k)})$ for some
function $f$. The $\operatorClassW$-hierarchy is a collection of computational complexity classes: we omit the technical
definitions here. The following relation is known amongst the classes in the $\operatorClassW$-hierarchy:
$\classFPT=\classW{0}\subseteq W[1]\subseteq W[2]\subseteq \ldots$. It is widely believed that $FPT\neq W[1]$, and hence if a
problem is hard for the class $W[i]$ (for any $i\geq 1$) then it is considered to be fixed-parameter intractable.

\medskip
\noindent\textbf{Previous Results:} Bhawalkar et al.~\cite{BhawalkarKLRS12} initiated the algorithmic  study of
\textsc{Anchored $k$-Core}  on undirected graphs and obtained an interesting dichotomy result:  the decision version of the
problem is solvable in polynomial time for $k\leq 2$ and is  \classNP-complete for all $k\geq 3$. For $k\geq 3$, they also
studied the problem from the viewpoint of parameterized complexity and approximation algorithms.
%They also proved that for every $k\geq 3$, the problem is \classW{2}-hard with respect to the budget parameter $b$. They also
%showed that the problem is \classFPT\  parameterized by the treewidth of the graph and proved some inapproximability results
%for $k\geq 3$.
%PG:
The current set of authors~\cite{ChitnisFG13} improved and generalized these results by showing that for $k\geq 3$ the problem
remains \classNP-complete even on planar graphs.

\medskip\noindent\textbf{Our Results:} In this paper we provide a number of new results on the algorithmic complexity of
\textsc{Directed Anchored $k$-Core}  (\DAKC).
%
%First observe that the undirected version of \textsc{Anchored $k$-Core} problem can be modeled by the directed version: simply
%replace each edge $\{u,v\}$ by arcs $(u,v)$ and $(v,u)$. Keeping the parameters $b,k,p$ unchanged it is now easy to see that
%the two instances are equivalent.
%
We start (Section~\ref{sec:defs}) by showing that that the decision version of \DAKC{} is \classNP-complete for every $k\geq
1$  even if the input graph is restricted to be a planar directed acyclic graph (DAG) of maximum degree at most $k+2$. Note
that this shows that the directed version is in some sense strictly harder than the undirected version since it is known be in
\classP\ if $k\leq 2$, and \classNP-complete if $k\geq 3$~\cite{BhawalkarKLRS12}. The \classNP-hardness result for \DAKC
motivates us to make a more refined analysis of the \DAKC problem via the paradigm of parameterized complexity. In
Section~\ref{sec:saved}, we  obtain the following dichotomy result: \DAKC is \classFPT \, parameterized by $p$ if $k=1$, and
\classW1-hard if $k\geq 2$. This fixed-parameter intractability result parameterized by $p$ forces us to consider the
complexity on special classes of graphs such as bounded-degree directed graphs or directed acyclic graphs. In
Section~\ref{sec:bound-deg}, for graphs of degree upper bounded by $\Delta$, we show that the \DAKC problem is FPT
parameterized by $p+\Delta$ if $k\geq \frac{\Delta}{2}$. In particular, it implies that \DAKC is FPT parameterized by $p$ for
directed graphs of maximum degree at most four. We complement these results by showing in Section~\ref{sec:concl} that  if $k<
\frac{\Delta}{2}$ and $\Delta\geq 3$, then \DAKC is \classW2-hard when parameterized by the number of anchors $b$ even for
DAGs, but the problem is \classFPT\ when parameterized by $\Delta+p$ for DAGs of maximum degree at most $\Delta$. Note that we
can always assume that $b\leq p$, and hence any \classFPT\ result with parameter $b$ implies \classFPT\ result with parameter
$p$ as well. On the other side, any hardness result with respect to $p$ implies the same hardness with respect to $b$.

\section{Preliminaries}\label{sec:defs}
We consider finite directed and undirected graphs without loops or multiple arcs. The vertex set of a (directed) graph $G$ is
denoted by $V(G)$ and its edge set (arc set for a directed graph) by $E(G)$. The subgraph of $G$ induced by a subset
$U\subseteq V(G)$ is denoted by $G[U]$. For $U\subset V(G)$ by $G-U$ we denote the graph $G[V(G)\setminus U]$. For a directed
graph $G$, we denote by $G^*$ the undirected graph with the same set of vertices such that $\{u,v\}\in E(G^*)$ if and only if
$(u,v)\in E(G)$. We say that $G^*$ is the \emph{underlying} graph of $G$.

Let $G$ be a directed graph. For a vertex $v\in V(G)$, we say that $u$ is an \emph{in-neighbor} of $v$ if $(u,v)\in E(G)$. The
set of all in-neighbors of $v$ is denoted by $N_G^-(v)$. The \emph{in-degree} $d_G^-(v)=|N_G^-(v)|$. Respectively, $u$ is an
\emph{out-neighbor} of $v$ if $(v,u)\in E(G)$, the set of all out-neighbors of $v$ is denoted by $N_G^+(v)$, and the
\emph{out-degree} $d_G^+(v)=|N_G^+(v)|$. The \emph{degree} $d_G(v)$ of a vertex $v$ is the sum $d_G^-(v)+d_G^+$, and the
\emph{maximum degree} of $G$ is $\Delta(G)=\max_{v\in V(G)}d_G(v)$. A vertex $v$ of $d_G^-(v)=0$ is called a \emph{source},
and if  $d_G^+(v)=0$, then $v$ is a \emph{sink}. Observe that isolated vertices are sources and sinks simultaneously.

Let $G$ be a directed graph. For $u,v\in V(G)$, it is said that $v$ can be \emph{reached} (or \emph{reachable}) from $u$ if
there is a directed $u\rightarrow v$ path in $G$. Respectively, a vertex $v$ can be reached from a set $U\subseteq V(G)$ if
$v$ can be reached from some vertex $u\in U$. Notice that each vertex is reachable from itself. We denote by $R_G^+(u)$
($R_G^+(U)$ respectively) the set of vertices that can be reached from a vertex $u$ (a set $U\subseteq V(G)$ respectively).
Let $R_G^-(u)$ denote the set of all vertices $v$ such that $u$ can be reached from $v$.

For two non-adjacent vertices $s,t$ of a directed graph $G$, a set $S\subseteq V(G)\setminus\{s,t\}$ is said to be a
\emph{$s-t$ separator} if $t\notin R_{G-S}^+(s)$. An $s-t$ separator $S$ is \emph{minimal} if no proper subset $S'\subset S$
is a $s-t$ separator.

The notion of important separators was introduced by Marx~\cite{Marx06} and generalized for directed graphs
in~\cite{ChitnisHM12}. We need a special variant of this notion. Let $G$ be a directed graph, and let $s,t$ be non-adjacent
vertices of $G$. An minimal $s-t$ separator is an \emph{important $s-t$ separator} if there is no $s-t$ separator $S'$ with
$|S'|\leq |S|$ and $R_{G-S}^-(t)\subset R_{G-S'}^-(t)$. The following lemma is a variant of Lemma~4.1 of~\cite{ChitnisHM12}.
Notice that to obtain it, we should replace the directed graph in Lemma~4.1 of~\cite{ChitnisHM12} by the graph obtained from
it by reversing direction of all arcs.

\begin{lemma}[\cite{ChitnisHM12}]\label{lem:imp-sep}
Let $G$ be a directed graph with $n$ vertices, and let $s,t$ be non-adjacent vertices of $G$. Then for every $h\geq 0$, there
are at most $4^h$ important $s-t$ separators of size at most $h$. Furthermore, all these separators can be enumerated in time
$O(4^{h}\cdot n^{O(1)})$.
\end{lemma}

As further we are interested in the parameterized complexity of \DAKC, we show first \classNP-hardness of the problem. 
%(the proof is given in Appendix~\ref{a:NPc}).

\begin{theorem}\label{thm:NPc}
For any $k\geq 1$, \DAKC is \classNP-complete, even for planar DAGs of maximum degree at
most $k+2$.
\end{theorem}

\begin{proof}
We reduce {\sc Satisfiability}:
\begin{center}
\noindent\framebox{\begin{minipage}{4.50in}{Satisfiability}\\
\emph{Input}: Sets of Boolean variables $x_1,\ldots,x_n$ and clauses $C_1,\ldots,C_m$. \\
\emph{Question}: Can the formula $\phi=C_1\vee\ldots\vee C_m$ be satisfied?
\end{minipage}}
\end{center}
It is known (see e.g.~
\cite{DahlhausJPSY94}) that this problem remains \classNP-hard even if each clause contains at most 3 literals (notice that clauses of size one or two are allowed), each variable
is used in at most 3 clauses: at least once in positive and at least once in negation, and the graph that correspond to a
boolean formula is planar. Consider an instance of {\sc Satisfiability} with $n$ variables $x_1,\ldots,x_n$ and $m$ clauses
$C_1,\ldots,C_m$ that satisfies these restrictions on planarity and the number of occurrences of the variables. We construct the graph $G$
as follows.

\begin{figure}[ht]
\centering\scalebox{0.7}{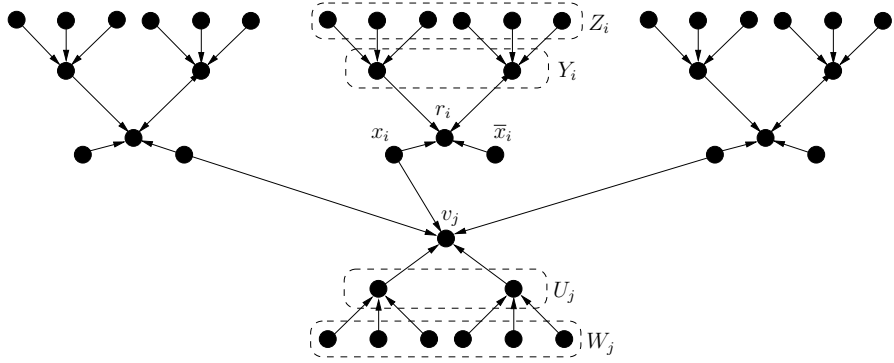}
\caption{Construction of $G$ for $k=3$.
\label{fig:NPh}}
\end{figure}

\begin{itemize}
\item For each $i\in\{1,\ldots,n\}$,
\begin{itemize}
\item add vertices $x_i,\overline{x}_i,r_i$ and add arcs $(x_i,r_i),(\overline{x}_i,r_i)$;
\item add a set of $k-1$ vertices $Y_i$ and draw an arc from each of them to $r_i$;
\item for each vertex $y\in Y_i$, add $k$ vertices and draw an arc from each of them to $y$, denote the set of these
    $k(k-1)$ vertices $Z_i$.
\end{itemize}
\item For each $j\in\{1,\ldots,j\}$,
\begin{itemize}
\item add a vertex $v_j$, and for each literal $x_i$ ($\overline{x}_i$ respectively) in the clause $C_j$, join the
    vertex $x_i$ ($\overline{x}_i$ respectively) with $v_j$ by an arc;
\item add a set of $k-1$ vertices $U_j$ and draw an arc from each of them to $v_j$;
\item for each vertex $u\in U_j$, add $k$ vertices and draw an arc from each of them to $u$, denote the set of these
    $k(k-1)$ vertices $W_j$.
\end{itemize}
\end{itemize}
Notice that if $k=1$, then $Y_i=Z_i=U_j=W_j=\emptyset$. The construction of $G$ is shown in Fig.~\ref{fig:NPh}. We set
$b=n(k(k-1)+1)+mk(k-1)$ and $p=n((k+1)(k-1)+2)+m((k+1)(k-1)+1)$. It is straightforward to see that $G$ is acyclic. Because
each variable $x_i$ is used at most 2 times in positive and at most 2 times in negations, $d_G(x_i), d_G(\overline{x}_i)\leq
3$ for all $i\in\{1,\ldots,n\}$, and $\Delta(G)\leq k+2$. 
%PG:
%Since $d_G(x_i)=1$ or $d_G(\overline{x}_i)=1$ for each $i\in\{1,\ldots,n\}$,
Because the graph of the boolean formula is a subcubic planar graph,  
$G$ is planar.

We claim that all clauses $C_1,\ldots,C_m$ can be satisfied if and only if there are a set $A\subseteq V(G)$
and an induced subgraph $H$ of $G$ such
that $A\subseteq V(H)$, $|A|\leq b, |V(H)|\geq p$, and for every $v\in V(H)\setminus A$, we have $d^{-}_{H}(v)\geq k$.

Suppose that we have a YES-instance of {\sc Satisfiability} and consider a truth assignment of $x_1,\ldots,x_n$ such that
all clauses are satisfied. We construct $A$ by including all the vertices $Z_1\cup\ldots\cup Z_n\cup W_1\cup\ldots\cup W_m$ in
this set, and for each $i\in\{1,\ldots,n\}$, if $x_i=\text{true}$, then $x_i$ is included in $A$ and $\overline{x}_i$ is
included otherwise. Clearly, $|A|=|Z_1|+\ldots+|Z_n|+|W_1|+\ldots+|W_m|+n=n(k(k-1)+1)+mk(k-1)=b$. Let $H=G[A\cup
Y_1\cup\ldots\cup Y_n\cup U_1\cup\ldots U_m\cup\{r_1,\ldots,r_n\}\cup\{v_1,\ldots,v_m\}]$. Consider $w\in V(H)\setminus A$. If
$w\in Y_i$ for $i\in\{1,\ldots,n\}$, then $w$ has $k$ in-neighbors in $Z_i\subseteq A$. If $w=r_i$ for $i\in\{1,\ldots,n\}$,
then $w$ has $k-1$ in-neighbors in $Y_i$ and either $x_i$ or $\overline{x}_i$ is an in-neighbor of $w$ as well. If $w\in U_j$
for $j\in\{1,\ldots,m\}$, then $w$ has $k$ in-neighbors in $W_j\subseteq A$. Finally, if $w=v_j$ for some
$j\in\{1,\ldots,m\}$, then $w$ has $k-1$ in-neighbors in $U_j$. As the clause $C_j$ is satisfied, it contains a literal $x_i$
or $\overline{x_i}$ that has the value $\text{true}$. Then by the construction of $A$, the corresponding vertex $x_i$ or
$\overline{x}_i$ respectively is in $A$, and $w$ has one in-neighbor in $A$. It remains to observe that
$|V(H)|=|A|+|Y_1|+\ldots+|Y_n|+|U_1|+\ldots+|U_m|=n(k(k-1)+1)+mk(k-1)+k(n+m)=p$.

Assume now there are a set $A\subseteq V(G)$
and an induced subgraph $H$ of $G$ such
that $A\subseteq V(H)$, $|A|\leq b, |V(H)|\geq p$ and for every $v\in V(H)\setminus A$ we have $d^{-}_{H}(v)\geq k$.

Let $S=\{w\in V(G)\ |\ d_G^-(w)=0\}=(\cup_{i=1}^n\{x_i,\overline{x_i}\})\cup (\cup_{i=1}^nZ_i)\cup(\cup_{j=1}^mW_j)$ and
%PG:
%$T=\{w\in V(G)\ |\ d_G^+(w)=0\}=V(G)\setminus S=\{r_1,\ldots,r_n\}\cup (\cup_{i=1}^nY_i)\cup(\cup_{j=1}^mU_j)$. 
$T=V(G)\setminus S=\{r_1,\ldots,r_n\}\cup (\cup_{i=1}^nY_i)\cup(\cup_{j=1}^mU_j)$. 
 We claim that
$A\subseteq S$ and $T\subseteq V(H)$. To show it, observe that any vertex $w\in S$ is in $H$ if and only if $w\in A$ as
$d_G^-(w)=0$. Because $|V(G)|-|V(H)|\leq n$, at least $|S|-n$ vertices of $S$ are in $A$. Since $|S|=b+n$, we conclude that
exactly $b=|S|-n$ vertices of $S$ are in $A$ and $A\subseteq S$. Moreover, $V(H)=T\cup A$.

Let $z\in Z_i$ for some $i\in\{1,\ldots,n\}$ and assume that $z$ is adjacent to $y\in Y_i$. If $z\notin A$, then $y\in T$ has
at most $k-1$ in-neighbors in $H$, a contradiction. Hence, $Z_1\cup\ldots\cup Z_n\subseteq A$. By the same arguments we
conclude that $W_1\cup\ldots\cup W_m\subseteq A$. Then we have exactly $n$ elements of $A$ in
$\cup_{i=1}^n\{x_i,\overline{x}_i\}$. Consider a pair of vertices $x_i,\overline{x}_i$ for $i\in\{1,\ldots,n\}$. If
$x_i,\overline{x}_i\notin A$, then $r_i\in T$ has at most $k-1$ in-neighbors in $H$, a contradiction. Therefore, for each
$i\in\{1,\ldots,n\}$, exactly one vertex from the pair  $x_i, \overline{x}_i$ is in $A$. For $i\in\{1,\ldots,n\}$, we set the
variable $x_i=\text{true}$ if the vertex $x_i\in A$, and $x_i=\text{false}$ otherwise.

It remains to prove that we have a satisfying truth assignment. Consider a clause $C_j$ for $j\in\{1,\ldots,m\}$. The vertex
$v_j\in T$ has $k-1$ in-neighbors in $H$ that are vertices of $T$. Hence, it has at least one in-neighbor in $A$. It can be
either a vertex $x_i$ or $\overline{x}_i$ that correspond to a literal in $C_j$. It is sufficient to observe that if $x_i\in
A$, then the literal $x_i=\text{true}$, and if  $\overline{x}_i\in A$, then the literal $\overline{x}_i=\text{true}$ by our
assignment.
\end{proof}

We conclude this section by the simple observation that \DAKC is in \classXP\ when parameterized by the number of anchors $b$.
For a directed graph $G$ with $n$ vertices, we can consider all the at most $n^b$ possibilities to
choose the anchors, and then recursively delete non-anchor vertices that have the in-degree at most $k-1$. Trivially, if we
obtain a directed graph with at least $p$ vertices for some selection of the anchors, we have a solution and otherwise we can
answer NO.

\section{\DAKC parameterized by the size of the core}
\label{sec:saved}
\vskip-2mm
In this section we consider the \DAKC problem for fixed $k$ when $p$ is a parameter and obtain the following dichotomy: If
$k=1$ then the \DAKC problem is FPT parameterized by $p$, otherwise for $k\geq 2$ it is \classW1-hard parameterized by $p$.

\begin{theorem}\label{thm:fpt-d-1}
For $k=1$, the \DAKC problem is solvable in time $2^{O(p)}\cdot n^2\log n$ on digraphs with $n$ vertices.
\end{theorem}

\begin{proof}
The proof is constructive, and we describe an \classFPT~ algorithm for the problem. Without loss of generality, we assume that
$b< p\leq n$.

We apply the following preprocessing rule reducing the  instance to  an acyclic graph. Let  $C_1,\ldots,C_r$ be strongly
connected components of $G$. By making use of Tarjan's algorithm~\cite{Tarjan72}, the sets  $C_1,\ldots,C_r$ can be found  in
linear time.  Let $R=R_G^{+}\Big(\bigcup_{i=1}^rV(C_i)\Big)$ be the set of vertices reachable from strongly connected
components. Then every $v\in R$ satisfies $d_{G[R]}^-(v)\geq 1$. If $b\geq p-|R|$, then we select in $V(G)\setminus R$ any
arbitrary $b'=p-|R|$ vertices $a_1,\ldots,a_{b'}$. In this case we  output the set of anchors $A=\{a_1,\ldots,a_{b'}\}$ and
graph $H=G[A\cup R]$. Otherwise, if $b< p-|R|$,  we set $G'=G-R$ and $p'=p-|R|$ and consider a new instance of
\DAKC  with the graph $G'$ and the parameter $p'$.

To see that the rule is safe, it is sufficient to observe that a set of anchors $A$ and a subgraph $H'$ of size at least $p'$
is a solution of the obtained instance if and only if $(A,H=G[V(H')\cup R])$ is a solution for the original problem. Let us
remark that the preprocessing rule can be easily  performed in  time $O(n^2)$.

From now we can assume that $G$ has no strongly connected components, i.e., $G$ is a directed acyclic graph. Denote by
$S=\{s_1,\ldots,s_h\}$ the set of sources of $G$. If $|S|\leq b$, then set $A=S$. In this case, we output the pair $(A,H=G)$.
The pair $(A,H)$ is a solution because every vertex $v\in V(G)\setminus S$ satisfies $d_G^-(v)\geq 1$. It remains to consider
the case when $|S|>b$. For $i\in\{1,\ldots,h\}$, let $R_i=R_G^+(s_i)$. Then $V(G)=R_G^+(S)=\bigcup_{i=1}^hR_i$. Without loss
of generality, we can assume that every anchored vertex is from $S$. Indeed, if $s_i$ is an anchor, then each vertex of $R_i$
can be included in a solution. Hence for every anchor $a\in R_j\setminus \{s_j\}$, we can delete anchor from $a$ and anchor
$s_j$, if it is not yet anchored.
Since we  can choose anchors only from $S$, we are able to reduce the problem to {\sc Partial Set Cover}.

\begin{center}
\noindent\framebox{\begin{minipage}{4.50in}{\psc }\\
\emph{Input }: A collection $X=\{X_1,\ldots,X_r\}$ of subsets of a finite $n$-element set $U$ and positive integers $p, b$. \\
\emph{Parameter}: $p$.\\
\emph{Question}: Are there at most $b$ subsets $X_{i_1},\ldots,X_{i_{b}}$, $1\leq i_1<\ldots<i_{b}\leq r$, covering at least
$p$ elements of $U$, i.e., $|\bigcup_{j=1}^{b} X_{i_j}|\geq p$?
\end{minipage}}
\end{center}

Bl\"{a}ser~\cite{Blaser03} showed that \psc if \classFPT parameterized by $p$ and can be solved in time $O(2^{O(p)}\cdot r
n\log n)$.
For \DAKC, 
we consider the collection of subsets $\{R_1,\ldots,R_r\}$ of $V(G)$. If we can select at most $b$
subsets $R_{i_1},\ldots,R_{i_{b}}$ such that  $|\cup_{j=1}^{b}R_{i_j}|\geq p$, we return the solution with anchors
$A=\{s_{i_1},\ldots,s_{i_{b}}\}$ and $H=G[\bigcup_{j=1}^{b}R_{i_j}]$. Otherwise, we return a NO-answer.

Because our preprocessing can be done in time $O(n^2)$ and  {\sc Partial Set Cover} is solvable in time $2^{O(p)}\cdot n^2\log
n$, we conclude that the total running time is $2^{O(p)}\cdot n^2\log n$.
\end{proof}

Now we complement Theorem~\ref{thm:fpt-d-1} by showing that for $k\geq 2$,  \DAKC becomes hard
parameterized by the core size.

\begin{theorem}\label{thm:w-saved}
For any fixed $k\geq 2$, the \DAKC problem is \classW1-hard parameterized by $p$, even for DAGs.
\end{theorem}

\begin{proof}
We reduce from the {\sc $b$-Clique} problem which is known to be \classW1-hard~\cite{downey-fellows-book}:

\begin{center}
\noindent\framebox{\begin{minipage}{4.50in}{$b$-Clique}\\
\emph{Input}: A undirected graph  $G$ and a positive integer $b$. \\
\emph{Parameter}: $b$\\
\emph{Question}: Is there a clique of size $b$ in $G$?
\end{minipage}}
\end{center}

From a given graph $G$ we construct a directed graph $G'$ as follows.
\begin{itemize}
\item Construct a copy of $V(G)$.
\item For each edge $\{u,v\}\in E(G)$, construct a new vertex $w_{uv}$ and join $u,v$  with $w_{uv}$  the copy of $V(G)$
    by arcs $(u,w_{uv})$ and $(v,w_{uv})$.
\item Construct $k-2$ vertices $z_1,\ldots,z_{k-2}$, and for each $e\in E(G)$, join $z_1,\ldots,z_{k-2}$ with $w_e$ by
    arcs.
\end{itemize}
It is straightforward to see that $G'$ is a directed acyclic graph. We say that the vertex $w_{uv}$ for $\{u,v\}\in E(G)$ is a
\emph{subdivision} vertex, and we say that $v\in V(G)$  is a \emph{branch} vertex. Let $b'=b+k-2$ and
$p=\frac{b(b+1)}{2}+k-2$. Let $Z=\{z_1\ldots,z_{k-1}\}$. We claim that $G$ has a clique of size $b$ if and only if there is a
set of at most $b'$ vertices $A\subseteq V(G')$ such that there exists an an induced subgraph $H$ of $G'$ with at least $p$
vertices, $A\subseteq V(H)$ and for any $v\in V(H)\setminus A$ we have  $d_H^-(v)\geq k$.

Suppose that $K$ is a clique in $G$ of size $b$. We let $A=K\cup Z$ and define $U=\{w_{uv}|u,v\in K\}$. Notice that
$|U|=\frac{b(b-1)}{2}$ and each vertex of $U$ has two in-neighbors in $A\cap K$ and $k-2$ in-neighbors in $Z$. We conclude
that $H=G'[A\cup U]$ has $p$ vertices and  for any $v\in V(H)\setminus A$ satisfies $d_H^-(v)\geq k$.

Assume now that there is a set of at most $b'$ vertices $A\subseteq V[G']$ such that there exists an induced subgraph $H$ of
$G'$ with at least $p$ vertices, $A\subseteq V(H)$ and for any $v\in V(H)\setminus A$ we have  $d_H^-(v)\geq k$. Since
subdivision vertices of $G'$ are sinks, we can assume that $A$ contains only \emph{branch} vertices and vertices from $Z$, as
otherwise we can replace an anchor $a\in A$ that is a subdivision vertex of $G'$ by an arbitrary branch vertex or a vertex of
$Z$. Because branch vertices of $G'$ and the vertices of $Z$ are sources, any such vertex $v$ is in $H$ if and only if $v\in
A$. Hence, $H$ has at most $b'$ sources of $G'$ and at least $\frac{b(b-1)}{2}$ subdivision vertices. If there is a vertex
$z_i\in Z$ such that $z_i\notin A$, then each subdivision vertex $w_e$ has at most $k-1$ in-neighbors and $H$ cannot contain
subdivision vertices. Therefore $Z\subseteq A$ and $A$ has at most $b'-(k-2)=b$ branch vertices. It remains to observe that a
subdivision vertex $w_{uv}$ has $k$ in-neighbors in $H$ if and only if $u,v\in A$. Then the claim follows.
\end{proof}

\section{\DAKC on graphs of bounded degree}
\label{sec:bound-deg} 
In this section we show that \DAKC problem is \classFPT\ parameterized by $\Delta+p$ if
$k\geq\frac{\Delta}{2}$.

In our algorithms we need to check the existence of solutions for \DAKC that have bounded size. It can be observed that if we
are interested in solutions $(A,H)$ such that $p\leq |V(H)|\leq q$, then for every positive $q$, we can express this problem
in the first order logic. It was proved by Seese~\cite{Seese96} that any graph problem expressible in the first-order logic
can be solved in linear time on (directed) graphs of bounded degree. Later this result was extended for much more rich graph
classes (see~\cite{DvorakKT10} ). These meta theorems are very general, but do not provide
good upper bounds for running time for particular problems. Hence, we give the following lemma.
Our algorithms use the random separation technique due to Cai et al.~\cite{CaiCC06} (which is a variant of the color coding
method introduced by Alon et al.~\cite{AlonYZ95}) .

\begin{lemma}\label{lem:bounded}
There is a randomized algorithm with running time $2^{O(\Delta q)}\cdot n$ that for an instance of \DAKC with an $n$-vertex
directed graph of maximum degree at most $\Delta$ and a positive integer $q\geq p$, either returns a solution $(A,H)$ with $V(H)\geq p$ or gives the answer
that there is no solution with $|V(H)|\leq q$. Furthermore, the algorithm can be derandomized, and the deterministic variant
runs in time  $2^{O(\Delta q)}\cdot n\log n$.
\end{lemma}

\begin{proof}
Consider an instance of \DAKC with an $n$-vertex directed graph $G$ of maximum degree at most $\Delta$. We assume that $b\leq
p\leq n$. For given $q\geq p$, to decide if $G$ contains a solution of size at  most $q$, we do the following.

We color each vertex of $G$ uniformly at random with probability $\frac{1}{2}$ by one of  two colors, say red or
blue. Let  $R$ be the set of vertices  colored red. Observe that if there is a solution $(A,H)$ with $|V(H)|\leq q$, then with
probability at least $\frac{1}{2^q}$ all vertices of $H$ are colored red and with probability at least $\frac{1}{2^{\Delta q}}$
all 
in- and out-neighbors of the vertices of $H$ that are outside of $H$ are colored blue. Using this observation, we assume that
$H$ is the union of some weakly connected components of  the graph $G[R]$ induced by red vertices.

In time $O(\Delta n)$ we find all weakly connected components of $G[R]$. If there is a component $C$ with at least $b+1$
vertices of in-degree at most $k-1$ (in $C$), then we discard this component as it cannot be a part of any solution. Denote by
$C_1,\ldots,C_r$ the remaining components. For $i\in\{1,\ldots,r\}$, let $A_i=\{v\in V(C_i)|d_{C_i}^-(v)<k\}$, $b_i=|A_i|$ and
$p_i=|V(C_i)|$.

Thus everything boils down to the problem of finding a set $I\subseteq \{1,\ldots,r\}$ such that $\sum_{i\in I} b_i\leq b$ and
$\sum_{i\in I}p_i \geq p$.  But this is  the well known {\sc Knapsack} problem, which  is solvable in time $O(bn)$ by dynamic
programming. If we obtain a solution $I$, then we output $(A,H)$, where $A=\cup_{i\in I} A_i$ and $H=G[\cup_{i\in I}V(C_i)]$.
Otherwise, we return a NO-answer. Notice that this algorithm can also find a solution $(A,H)$  with $|V(H)|> q \geq p$.

It remains to observe that for any positive number $\alpha<1$, there is a constant $c_{\alpha}$ such that after running our
randomized algorithm $c_{\alpha}\cdot 2^{\Delta q}$ times, we either find a solution $(A,H)$ or can claim that with
probability $\alpha$ that it does not exist.

This algorithm can be derandomized by the  technique proposed by Alon et al.~\cite{AlonYZ95}: replace the random colorings by
a family of at most $2^{O(\Delta q)}\cdot \log n$ hash functions which are known to be constructible  in time $2^{O(\Delta
q)}\cdot n\log n$.
\end{proof}

Our next aim is to prove that for $k>\Delta/2$ the \DAKC problem is \classFPT\ when parameterized by the number of anchors
$b$.

\begin{lemma}\label{lem:bound-deg-anchors}
Let $\Delta$ be a positive integer. If $k>\Delta/2$, then the \DAKC problem can be solved in time $2^{O(\Delta^2 b)}\cdot
n\log n$ for $n$-vertex directed graphs of maximum degree at most $\Delta$.
\end{lemma}

\begin{proof}
Suppose $(A,H)$ is a solution for the \DAKC problem. Let us observe that because $k>\Delta/2$, for every vertex $v\in
V(H)\setminus A$,  we have $d_H^-(v)>d_H^+(v)$. Recall that for any directed graph, the sum of in-degrees equals the sum of
out-degrees. Then
$$  \sum_{v\in V(H)\setminus A}(d_H^-(v)-d_H^+(v))=\sum_{v\in A}(d_H^+(v)-d_H^-(v)).$$
Since for every vertex $v\in V(H)\setminus A$,  $d_H^-(v)-d_H^+(v)\geq 1$, we have that
$$|V(H)\setminus A|\leq \sum_{v\in V(H)\setminus A}(d_H^-(v)-d_H^+(v)).$$
On the other hand,  $ d_H^+(v)-d_H^-(v)\leq \Delta$, and we arrive at
$$|V(H)\setminus A|\leq \sum_{v\in V(H)\setminus A}(d_H^-(v)-d_H^+(v))=\sum_{v\in A}(d_H^+(v)-d_H^-(v))\leq \Delta|A|.$$
Hence, $|V(H)|\leq (\Delta+1)|A|\leq (\Delta+1)b$. Using this observation, we can solve the \DAKC problem as follows. If
$p>(\Delta+1)b$, then we return a NO-answer. If $p\leq(\Delta+1)b$, we apply Lemma~\ref{lem:bounded} for $q=(\Delta+1)b$, and
solve that problem in time  $2^{O(\Delta^2 b)}\cdot n\log n$.
\end{proof}

Now we show that if $k=\frac{\Delta}{2}$ then the \DAKC problem is FPT parameterized by $\Delta+p$. Due the space restrictions we only sketch the proof of the following lemma.

\begin{lemma}\label{lem:bound-deg-saved}
Let $\Delta$ be a positive integer. If $k=\Delta/2$, then the \DAKC problem can be solved in time $2^{O(\Delta^3 b+ \Delta^2 b
p)}\cdot n^{O(1)}$ for $n$-vertex directed graphs of maximum degree at most $\Delta$.
\end{lemma}

\begin{proof}
We describe an \classFPT~ algorithm. Consider an instance of the \DAKC problem. Without loss of generality we assume that $b<p\leq n$.

We apply the following preprocessing rule. Suppose that $G$ has a (weakly) connected component $C$ such that for any $v\in
V(C)$, $d_C^-(v)=d_C^+(v)=k$. If $b\geq p-|V(C)|$, then we choose a set $A$ of $b'=p-|V(C)|$ vertices arbitrary in
$V(G)\setminus V(C)$. Then we return a YES-answer, as the anchors $A$ and $H=G[A\cup V(C)]$ is a solution. Otherwise, if $b<
p-|V(C)|$, we let $G'=G-V(C)$ and $p'=p-|V(C)|$. Now we consider a new instance of the problem with the graph $G'$ and the
parameter $p'$. To see that the rule is safe, it is sufficient to observe that a set of anchors $A$ and a subgraph $H'$ of
size at least $p'$ is a solution of the obtained instance if and only if $A$ and $H=G[V(H')\cup V(C)]$ is a solution for the
original problem. From now we assume that $G$ has no such components.

\medskip

We need the following claim.

\medskip
\noindent
{\bf Claim A.} {\it
If an instance of the \DAKC problem has a core with at least $(\Delta p+1)b+1$ vertices, then it has a solution
$(A,H)$ with the following property: there is a vertex $t\in V(H)\setminus A$  reachable in $H$ from any vertex of $H$.
Moreover, for each vertex $v$ of $H$, there is a path from $v$ to $t$ with all vertices except $v$ in  $ V(H)\setminus A$.
}

\begin{proof}[Proof of Claim~A]
Let  $(A,H')$ be a solution with the set of anchors $A$ and such that $V(H')> (\Delta p+1)b$.

We show that $V(H')=R_{H'}^+(A)$, i.e., all vertices of $H'$ are reachable from the anchors. To obtain a contradiction,
suppose that there is a vertex $u\in V(H')$ such that $u\notin R_{H'}^+(A)$. Let $U=R_{H'}^-(u)$, i.e., $U$ is the set of
vertices from which we can reach $u$. Clearly, $A\cap U=\emptyset$. Therefore, $d_{H'}^-(v)\geq k=\Delta/2$ for $v\in U$.
Notice that for a vertex $v\in U$, $N_{H'}^{-}(v)\subseteq U$ by the definition. Hence, $d_{G[U]}^-(v)\geq k=\Delta/2$ for $v\in
U$. Because the sum of in-degrees equals the sum of out-degrees, for  every vertex $v\in U$, we have that
$d_{G[U]}^-(v)=d_{G[U]}^+(v)=k=\Delta/2$. Then $C=G[U]$ is a component of $G$ such that for every $v\in V(C)$,
$d_C^-(v)=d_C^+(v)=k$, but such components are excluded by the preprocessing; a contradiction.

Observe now that if $d_{H'}^-(v)<d_{H'}^+(v)$, then $d_{H'}^-(v)<k$ and thus $v\in A$. Hence, by adding at most $\Delta b$
(maybe multiple) arcs from $V(H')\setminus A$ to $A$, joining the vertices $v\in V(H')$ of degrees $d_{H'}^-(v)>d_{H'}^+(v)$ with
vertices of degrees  $d_{H'}^-(v)<d_{H'}^+(v)$, we  can  transform $H'$ into a disjoint union of  directed Eulerian graphs.
Since $V(H')=R_{H'}^+(A)$, each of these directed Eulerian graphs contains at least one vertex of $A$. Thus  the set of arcs
of $H'$ can be covered by at most $\Delta b$ arc-disjoint directed walks,   each walk starting from a vertex of $A$ and never
coming back to $A$. Because $d_{H'}^-(v)\geq k$ for $v\in V(H')\setminus A$, we have that $|E(G')|\geq k(|V(H')|-b)> \Delta
kbp$. Then there is a walk $W$ with at least $kp+1$ arcs.  Let $a\in A$ be the first vertex of $W$ and let $t$ be the last
vertex of the walk. The walk $W$ visits $a$ only once, $t$  and all other vertices of $W$ are visited at most $k$ times. We
conclude  that $W$ has at least $p$ vertices.

Let $R=R_{H'-A}^-(t)$ and let $A'=\{a\in A\ |\ N_{H'}^+(a)\cap R\neq\emptyset\}$. Consider $H=G[R\cup A']$. Since
$V(W)\subseteq V(H)$, $|V(H)|\geq p$.  For any $v\in V(H)\setminus A$, the in-neighbors of $v$ in $H'$ are in $H$ by the
construction and, therefore, $d_H^-(v)\geq k$. It remains to observe that to select at most $b$ anchors, we take $A'\subseteq
V(H)$.
\end{proof}

Using  Claim~A, we proceed with our algorithm. We try to find a solution such that $H$ has at most $q=(\Delta
p+1)b$ vertices by applying  Lemma~\ref{lem:bounded}. It takes time $O(2^{O(\Delta^2 bp)}\cdot n\log n)$. If we obtain a
solution, then we return it and stop. Otherwise, we conclude that every core  contains at least $(\Delta p+1)b+1$
vertices. By Claim~A, we can search for a solution   $H$ with a  non-anchor vertex $t$ which is reachable from
all other vertices of $H$ by directed paths  avoiding $A$. Notice that since $t$ is a non-anchor vertex, we have that
$d_G^-(t)\geq k$. We try  at most $n$ possibilities for all possible choices of  $t$, and solve our problem for each choice.
Clearly, if we get a YES-answer for one of the choices, we return it and stop. Otherwise, if we fail,  we return a NO-answer.

From now we assume that we already selected $t$. We denote by $G'$ the graph obtained from $G$ by adding  an artificial source
vertex $s$ joined by arcs with all the vertices $v\in V(G)$ with $d_G^-(v)<k$. Observe that $(s,t)\notin E(G')$.

\medskip

Suppose that $(A,H)$ is a solution with the set of anchors $A$ such that $t\in V(H)\setminus A$ is reachable in $H$ from any
vertex of $H$ by a path with all inner vertices in $V(H)\setminus A$. Denote by $\delta_{G'}(H)$ the set $\{v\in V(H)\ |\
N_{G'}^-(v)\setminus V(H)\neq \emptyset\}$, i.e., $\delta_{G'}(H)$ contains vertices that have in-neighbors outside $H$. We
need a chain of claims about the structure of $H$ in $G'$.

\medskip
\noindent
{\bf Claim~B.} {\it
$|\delta_{G'}(H)\setminus A|\leq \Delta b$.
}

\begin{proof}[Proof of Claim~B]
Let $X=\{v\in V(H)\ |\ d_H^-(v)\geq k\text{ and }d_H^+(v)<k\}$, $Y=\{v\in V(H)\ |\ d_H^-(v)=d_H^+(v)=k\}$ and $Z=\{v\in V(H)\
|\ d_H^-(v)<k\}$. Clearly, $$\sum_{v\in X}(d_H^-(v)-d_H^+(v))+ \sum_{v\in Y}(d_H^-(v)-d_H^+(v))=\sum_{v\in
Z}(d_H^+(v)-d_H^-(v))$$ Observe that $d_H^-(v)-d_H^+(v)\geq 1$ for $v\in X$,  $d_H^-(v)-d_H^+(v)=0$ for $v\in Y$ and
$d_H^+(v)-d_H^-(v)\leq \Delta$ for $v\in Z$. Hence, $|X|\leq \Delta|Z|$. If $d_H^-(v)<k$ for $v\in V(H)$, then $v\in A$. It
follows that $Z\subseteq A$ and $|Z|\leq b$. We have $|X|\leq \Delta b$. Consider a vertex $v\in \delta_{G'}(H)\setminus A$.
It has at least one in-neighbor outside $H$ in $G$ and $d_H^-(v)\geq k$. Then $d_H^+(v)<k$ and  $v\in X$. We conclude that
$\delta_{G'}(H)\setminus A\subseteq X$ and $|\delta_{G'}(H)\setminus A|\leq \Delta b$.
\end{proof}

\medskip
\noindent
{\bf Claim~C.} {\it
There is  an $s-t$ separator $S$ in $G'$ of size at most $(\Delta(k-1)+1)b$ such that
 $V(H)\setminus A\subseteq R_{G'-S}^-(t)$.
}

\begin{proof}[Proof of Claim~C]
Let $S=\Big(\delta_{G'}(H)\cap A\Big)\cup\Big(\bigcup_{v\in \delta_{G'}(H)\setminus A}(N_G^-(v)\setminus V(H)\Big)$, i.e., the
set containing all anchors that are in $\delta_{G'}$, and for each non-anchor vertex of $\delta_{G'}$ containing all its
in-neighbors outside of $H$. Consider a directed $(s,t)$-path $P$ in $G'$. Let $v$ be the first vertex in $P$ that is in
$V(H)$ and let $u$ be its predecessor in $P$. If $v\in A$, then $v\in S$. If $v\notin A$, then $u\neq s$ as $H$ has no
non-anchor vertices with in-degree at most $k-1$ in $G$. Then $u\in S$. We conclude that each $(s,t)$-path contains a vertex
of $S$, i.e., this set is an $s-t$ separator.

Observe that $V(H)\setminus A\subseteq R_{G'-S}^-(t)$ by the definition of $S$ and the fact that $t$ can be reached from any
vertex of $H$ in this graph by a path with all inner vertices in $V(H)\setminus A$.

It remains to show that $|S|\leq (\Delta(k-1)+1)b$.  By Claim~B, $|\delta_{G'}(H)\setminus A|\leq \Delta b$. A
vertex $v\in \delta_{G'}(H)\setminus A$ has at least one out-neighbor in $H$ because $t$ is reachable from $v$. Then $v$ has
at most $k-1$ in-neighbors outside $H$. Hence $|S|\leq |A|+(k-1)(\delta_{G'}(H)\setminus A)\leq (\Delta(k-1)+1)b$.
\end{proof}

Now we can prove the following claim about important $s-t$ separators in $G'$.

\medskip
\noindent
{\bf Claim~D.} {\it
There is an important $s-t$ separator $S^*$ of size at most $(\Delta(k-1)+1)b$ in $G'$ such that $V(H)\subseteq
R_{G'-S^*}^-(t)\cup S^*$.
}

\begin{proof}[Proof of Claim~D]
By Claim~C, there is an $s-t$ separator $S'$ in $G'$ of size at most $(\Delta(k-1)+1)b$ such that $V(H)\setminus
A\subseteq R_{G'-S'}^-(t)$. Notice that $S'$ not necessary a minimal separator, but there is a minimal $s-t$ separator
$S\subseteq S'$. Clearly, $|S|\leq (\Delta(k-1)+1)b$.

We show that  $V(H)\subseteq R_{G'-S}^-(t)\cup S$.  
Because $R_{G'-S'}^-(t)\subseteq  R_{G'-S}^-(t)$, we have that $V(H)\setminus A\subseteq R_{G'-S}^-(t)$. Also if an anchor $a$
is in  $R_{G'-S'}^-(t)$, then $a\in R_{G'-S}^-(t)$. Let $a\in A\cap S'$. If $a\in A\cap S$, then $a\in R_{G'-S}^-(t)\cup S$.
If  $a\notin S$, then by Claim~C, $a$ has an out-neighbor $v\in R_{G'-S'}^-(t)$ and in this case we have $a\in
R_{G'-S}^-(t)$.

It remains to observe that there is an important $s-t$ separator $S^*$ such that $|S^*|\leq |S|\leq (\Delta(k-1)+1)b$ and
$R_{G'-S}^-(t)\subseteq   R_{G'-S^*}^-(t)$. Therefore, $V(H)\subseteq R_{G'-S}^-(t)\cup S\subseteq R_{G'-S^*}^-(t)\cup S^*$.
\end{proof}

The next step of our algorithm is to check all important $s-t$ separators in $G'$ of size at most  $(\Delta(k-1)+1)$. By
Lemma~\ref{lem:imp-sep}, there are at most $4^{(\Delta(k-1)+1)b}$ important $s-t$ separators and they can be listed in time
$2^{O(\Delta^2b)}\cdot n^c$. For each important $s-t$ separator $S^*$, we consider the set of vertices $U= R_{G'-S^*}^-(t)\cup
S^*$ and decide whether there is a solution such that $V(H)\subseteq U$. If we have a solution for some $S^*$, then we return
a YES-answer and stop. Otherwise, if we fail  to find such a solution for all important separators, we use
Claim~D  to deduce that there is no solution.

\medskip
From now on, we assume that an important $s-t$ separator $S^*$ is given and that  $U= R_{G'-S^*}^-(t)\cup S^*$.
In what follows, we describe a procedure of finding  a solution with  $V(H)\subseteq U$.

\medskip
Denote by $D$ the set $\{v\in U\ |\ d_G^-(v)>0\}$. We need the following observation.

\medskip
\noindent
{\bf Claim~E.} {\it
Set $D$ contains at most $(\Delta+1) (\Delta(k-1)+1)b$ vertices.
}

\begin{proof}[Proof of Claim~E]
Let $Q=G[U]$. Let $X=\{v\in V(Q)\ |\ d_Q^-(v)\geq k\text{ and }d_Q^+(v)<k\}$, $Y=\{v\in V(Q)\ |\ d_Q^-(v)=d_Q^+(v)=k\}$ and
$Z=\{v\in V(Q)\ |\ d_Q^-(v)<k\}$. Clearly,
$$\sum_{v\in X}(d_Q^-(v)-d_Q^+(v))+ \sum_{v\in Y}(d_Q^-(v)-d_Q^+(v))=\sum_{v\in Z}(d_Q^+(v)-d_Q^-(v))$$
Observe that $d_Q^-(v)-d_Q^+(v)\geq 1$ for $v\in X$,  $d_Q^-(v)-d_Q^+(v)= 0$ for $v\in Y$ and $d_Q^+(v)-d_Q^-(v)\leq \Delta$
for $v\in Z$. Hence, $|X|\leq \Delta|Z|$.

Recall that $G'$ is obtained from $G$ by joining $s$ with all vertices of in-degree at most $k-1$. Since $S^*$ is an $s-t$
separator, if for $v\in U$, $d_Q^-(v)<k$, then $v\in S^*$. Hence, $Z\subseteq S^*$ and $|Z|\leq |S^*|\leq (\Delta(k-1)+1)b$.
If for for $v\in U$, $d_G^-(v)>k$, then $v\in X\cup Z$. We conclude that $|D|\leq|X|+|Z|\leq (\Delta+1)|Z|\leq (\Delta+1)
(\Delta(k-1)+1)b$.
\end{proof}

Recall   that  set $\delta_{G'}(H)$ contains vertices of $H$ that have in-neighbors outside of $H$. If $v\in
\delta_{G'}(H)\setminus A$, then it has at least $k$ in-neighbors in $H$ and at least one in-neighbor outside $H$. Notice that
$s\notin N_{G'}^-(v)$ because $d_G^-(v)\geq d_H^-(v)\geq k$. Hence, $d_G^-(v)>k$. Because $V(H)\subseteq U$,
$\delta_{G'}(H)\setminus A\subseteq D$. By Claim~C, $|\delta_{G'}(H)\setminus A|\leq \Delta b$, and by
Claim~E, $|D|\leq (\Delta+1) (\Delta(k-1)+1)b$.
We consider  all at most
$  2^{(\Delta+1) (\Delta(k-1)+1)b}$ possibilities to select  $\delta_{G'}(H)\setminus A$. For each choice of
$\delta_{G'}(H)\setminus A$, we guess the arcs that join the vertices that are outside $H$ with the vertices of
$\delta_{G'}(H)\setminus A$ and delete them. Denote the graph obtained from $G$ by $F$. Recall that from each vertex $v$ of
$\delta_{G'}(H)\setminus A$, there is a directed path to $t$ that avoids $A$. Hence, $v$ has at least one out-neighbor in $H$
and at most $\Delta-1$ in-neighbors in $G$. Also $v$ has at least $k$ in-neighbors in $H$, and we delete at most $d_G^-(v)-k$
arcs. Therefore, for $v$ we choose at most $k-1$ arcs out of at most $\Delta-1$ arcs. We can upper bound the number of
possibilities for $v$ by $2^{\Delta-1}$, and the total number of possibilities for $\delta_{G'}(H)\setminus A$ is
$2^{(\Delta-1)\Delta b}$.

Observe that $(A,H)$ is a solution for the new instance of \DAKC, where $G$ is replaced by $F$ for a correct guess of the
deleted arcs. Also each solution for the new instance provides a solution for the graph $G$, because if we put deleted arcs
back, then we can only increase in-degrees. Hence, we can check for each possible choice of the set of deleted arcs, whether
the new instance has a solution. If for some choice we obtain a solution, then we return a YES-answer. Otherwise, if we fail
for all choices, then we return a NO-answer. Further we assume that $F$ is given.

Denote by $F'$ the graph obtained from $F$ by the addition of a vertex $s$ joined by arcs with all the vertices $N_{G'}^+(s)$.
Now $\delta_{F'}(H)=\{v\in V(H)\ |\ N_{F'}^-(v)\setminus V(H)\neq \emptyset\}$. By the choice of $F$,
$\delta_{F'}(H)=\delta_{G'}(H)\cap A$ and, therefore, $|\delta_{F'}(H)|\leq b$. Also $\delta_{F'}(H)$ is an $s-t$ separator in
$F'$ by Claim~C.

Now we can prove the following.

\medskip
\noindent
{\bf Claim~F.} {\it
There is an important $s-t$ separator $\hat{S}$ of size at most $b$ in $F'$ such that $(\hat{S},G[R_{F'-\hat{S}}^-(t)\cup
\hat{S}])$ is a solution for the instance of the \DAKC problem for the graph $G$.
}

\begin{proof}[Proof of Claim~F]
Let $U=R_{F'-\hat{S}}^-(t)\cup \hat{S}$. It was already observed that $\delta_{G'}^*(H)$ is an $s-t$ separator in $F'$ of size
at most $b$. Then there is a minimal $s-t$ separator $S\subseteq \delta_{G'}^*(H)$. Clearly, $|S|\leq b$.

As before in the proof of Claim~D, we show that  $V(H)\subseteq R_{F'-S}^-(t)\cup S$. Because for any vertex
$v$ of $H$, there is a directed $(v,t)$ path with all inner vertices in $V(H)\setminus A$, $V(H)\setminus A\subseteq
R_{F'-\delta_{F'}(H)}^-(t)$. Because $R_{F'-\delta_{F'}(H)}^-(t)\subseteq  R_{F'-S}^-(t)$ we have $V(H)\setminus A\subseteq
R_{F'-S}^-(t)$. Also if  $a\in A$ is in  $R_{F'-\delta_{F'}(H)}^-(t)$, then $a\in R_{F'-S}^-(t)$. Let $a\in A\cap
\delta_{F'}(H)$. Trivially, if $a\in A\cap S$, then $a\in R_{F'-S}^-(t)\cup S$. If  $a\notin S$, then $a$ has an out-neighbor
$v\in R_{F'-\delta_{F'}(H)}^-(t)$ and $a\in R_{F'-S}^-(t)$. Then there is an important $s-t$ separator $\hat{S}$ such that
$|\hat{S}|\leq |S|\leq b$ and $R_{F'-S}^-(t)\subseteq R_{F'-\hat{S}}^-(t)$. Therefore, $V(H)\subseteq R_{F'-S}^-(t)\cup
S\subseteq R_{F'-S^*}^-(t)\cup S^*$, and $|U|\geq p$.

It remains to observe that $s$ is adjacent to all vertices of $G$ with in-degrees at most $k-1$ and $S^*$ is an $s-t$
separator. It immediately follows that for any  vertex $v\in R_{F'-S^*}^-(t)$, $d_{F(U)}^-(v)\geq k$. Then
$(\hat{S},G[R_{F'-\hat{S}}^-(t)\cup \hat{S}])$ is a solution.
\end{proof}

The final step of our algorithm is to enumerate all important $s-t$ separators $\hat{S}$ of size at most $b$ in $F'$, which
number by Lemma~\ref{lem:imp-sep} is at most $4^b$, and for each $\hat{S}$, check whether $(\hat{S},G[R_{F'-\hat{S}}^-(t)\cup
\hat{S}])$ is a solution. Recall that all these separators can be listed in time $2^{O(b)}\cdot n^c$. We return a YES-answer
if we obtain a solution for some important separator, and a NO-answer otherwise.

To complete the proof, let us observe that each step of the algorithm runs either in polynomial or \classFPT~ time.
Particularly, the preprocessing is done in time $O(\Delta n)$. Then we check the existence of a solution of a bounded size in
time $2^{O(\Delta^2 bp)}\cdot n\log n$. Further we consider at most $n$ possibilities to choose $t$. For each $t$, we consider
at most $4^{(\Delta(k-1)+1)b}$ important $s-t$ separators $S^*$. Recall, that they can be listed in time
$2^{O(\Delta^2b)}\cdot n^c$ for some constant $c$. Then for each $S^*$, we have  at most  $2^{(\Delta+1)
(\Delta(k-1)+1)b+(\Delta-1)}$ possibilities to construct $F$, and it can be done in time $2^{O(\Delta^3b)}+O(\Delta n)$.
Finally, there are at most $4^b$ important $s-t$ separators $\hat{S}$ and they can be listed in time $2^{O(b)}\cdot n$ for
some $c$. We conclude that the total running time is $2^{O(\Delta^3b+\Delta^2bp)}\cdot n^c$ for some constant $c$.
\end{proof}

Combining Lemmas~\ref{lem:bound-deg-anchors} and \ref{lem:bound-deg-saved}, we obtain the following theorem.

\begin{theorem}\label{thm:bound-deg}
Let $\Delta$ be a positive integer. If $k\geq \frac{\Delta}{2}$, then the \DAKC problem can be solved in time $2^{O(\Delta^3
b+ \Delta^2 b p)}\cdot n^{O(1)}$ for $n$-vertex directed graphs of maximum degree at most $\Delta$.
\end{theorem}

Theorems~\ref{thm:fpt-d-1} and \ref{thm:bound-deg} give the next corollary.

\begin{corollary}\label{cor:fpt-delta}
The \DAKC problem can be solved in time $2^{O(bp)}\cdot n^{O(1)}$  for $n$-vertex directed graphs of maximum degree at most
$4$.
\end{corollary}

\section{Conclusions}\label{sec:concl}
We proved that \DAKC is \classNP-complete even for planar DAGs of maximum degree at most $k+2$. It was also shown  that \DAKC
is \classFPT~ when parameterized by $p+\Delta$ for directed graphs of maximum degree at most
$\Delta$ whenever $k\geq\Delta/2$. It is natural to ask whether the problem is \classFPT~ for other values $k$. This question
is interesting even for the special case $\Delta=5$ and $k=2$.

For the special case of  directed acyclic graphs (DAGs) we understand  the complexity of the problem much better.
Theorem~\ref{thm:w-saved} showed that \DAKC on DAGs is \classW{1}-hard parameterized by $p$ for every fixed $k\geq 2$, when
the degree of the graph is not bounded. We now show the following theorem  that gives
\classW{2}-hardness of \DAKC when parameterized  by the number of anchors $b$ (recall that we can always assume that $b\leq
p$).

\begin{theorem}\label{thm:w-hardness-dags-bounded-degree}
For any $\Delta\geq 3$ and any positive $k<\frac{\Delta}{2}$,  \DAKC  is \classW{2}-hard (even on DAGs) when parameterized by
the number of anchors $b$ on graphs of maximum degree at most $\Delta$.
\end{theorem}

\begin{proof}
First, we prove the claim for $k=1$ and $\Delta=3$. We reduce from the {\sc $b$-Set Cover} problem which is known to be
\classW2-hard~\cite{downey-fellows-book}:

\begin{center}
\noindent\framebox{\begin{minipage}{4.50in}{\textsc{$b$-Set Cover} }\\
\emph{Input }: A collection $X=\{X_1,\ldots,X_r\}$ of subsets of a finite $n$-element set $U$ and a positive integer $b$. \\
\emph{Parameter}: $b$\\
\emph{Question}: Are there at most $b$ subsets $X_{i_1},\ldots,X_{i_{b}}$ such that these sets cover $U$, i.e., $U=
    \bigcup_{j=1}^{b}X_{i_j}$?
\end{minipage}}
\end{center}

\begin{figure}[ht]
\centering\scalebox{0.7}{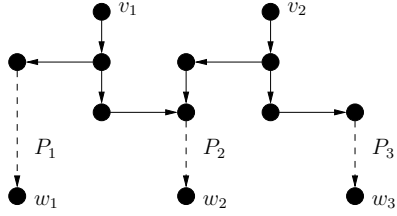}
\caption{Construction of $G$ for $U=\{u_1,u_2,u_3\}$ and $X_1=\{u_1,u_2\},X_2=\{u_2,u_3\} $.
\label{fig:W2h-1}}
\end{figure}

Let $U=\{u_1,\ldots,u_n\}$. We construct the directed graph $G$ as follows (see Fig.~\ref{fig:W2h-1}).
\begin{itemize}
\item For $i\in\{1,\dots,r\}$, assume that $X_i=\{u_{j_1},\ldots,u_{j_s}\}$ and
\begin{itemize}
\item construct a vertex $v_i$ and $s$ vertices $x_{ij_1},\ldots,x_{ij_s}$;
\item construct arcs $(v_i,x_{ij_1}),(x_{ij_1},x_{ij_2}),\ldots,(x_{ij_{s-1}},x_{ij_s})$.
\end{itemize}
\item For $j\in\{1,\ldots,n\}$, assume that $u_j$ is included in the sets $X_{i_1},\ldots,X_{i_t}$ and
\begin{itemize}
\item construct a vertex $w_j$ and $t$ vertices $y_{ji_1},\ldots,y_{ji_t}$;
\item construct arcs $(y_{ji_1},y_{ji_2}),\ldots,(y_{ji_{t-1}},y_{ji_t})$;
\item join $y_{ji_t}$ with $w_j$ by a directed path $P_j$ of length $\ell=2rn+r$.
\end{itemize}
\item For $i\in\{1,\dots,r\}$ and $j\in\{1,\ldots,n\}$, if $u_j\in X_i$, then construct an arc $(x_{ij},y_{ji})$.
\end{itemize}
It is straightforward to see that $G$ is a directed acyclic graph of maximum degree at most 3. We set $p=n \ell$. We claim
that $U$ can be covered by at most $b$ sets if and only if there is a set of at most $b$ vertices $A$ such that there exists
an induced subgraph $H$ of $G$ with at least $p$ vertices, $A\subseteq V(H)$ and for any $v\in V(H)\setminus A$, $d_H^-(v)\geq
1$.

Notice that $v_1,\ldots,v_r$ are the sources of $G$, $w_1,\ldots,w_n$ are the sinks, and $V(G)=\bigcup_{i=1}^rR_{G}^+(v_i)$.
Observe also that $w_j$ can be reached from $v_i$ if and only if $u_j\in X_i$.

Suppose that $U$ can be covered by at most $b$ sets say $X_{i_1},\ldots,X_{i_{b}}$. Let $A=\{v_{i_1},\ldots,v_{i_{b}}\}$ and
$H=G[R_{G}^+(A)]$. It is straightforward to see that for any vertex $z\in V[H]$, $d_{H}^-(z)\geq 1$. Because $U$ is covered,
all vertices $w_1,\ldots,w_n$ are in $H$ and, therefore, $V(P_1)\cup\ldots\cup V(P_n)\subseteq V(H)$. It remains to observe
that $|V(P_1)\cup\ldots\cup V(P_n)|=n(\ell+1)\geq p$ and we conclude that $(A,H)$ is a solution of our instance of \DAKC.

Assume now that $(A,H)$ is a solution of the \DAKC problem. Without loss of generality we can assume that that each $a\in A$
is a source of $G$. Otherwise, $a\in R_G^+(v_i)$ for some source $v_i$, and we can replace $a$ by $v_i$ in $A$ (or delete it
if $v_i\in A$ already). Let $\{i\ |\ 1\leq i\leq n,~v_i\in A\}=\{i_1,\ldots,i_{b}\}$. We show that $X_1,\ldots,X_{i_{b}}$
cover $U$. To obtain a contradiction, assume that there is an element $u_j\in U$ such that $u_j\notin X_{i_1}\cup \ldots\cup
X_{i_{b}}$. Then the vertex $w_j$ is not reachable from $A$. Hence, the vertices of $P_j$ are not reachable from $A$. It
follows that $V(P_j)\cap V(H)=\emptyset$. We have that $|V(H)|\leq |V(G)|-|V(P_j)|$. Because $|X_i|\leq n$ for
$i\in\{1,\ldots,r\}$ and each $u_h$ is included in at most $r$ sets for $h\in\{1,\ldots,n\}$, $|V(G)|\leq
r(n+1)+n(r+\ell)=2rn+r+n \ell=2rn+r+p$. Therefore, $|V(H)|\leq p+(2rn+r-(\ell+1))<p$ because $P_j$ has $\ell+1$ vertices; a
contradiction.

Now we prove \classW{2}-hardness for $k\geq 2$ and $\Delta>2k$. We reduce from an instance of the \DAKC problem with $k=1$
and $\Delta=3$. Consider an instance of this problem with a directed acyclic graph $G$ and positive integers $b, p$. Assume
that $b\leq p\leq |V(G)|$ and $|V(G)|\geq 3$. We construct the graph $G'$ as follows (see Fig.~\ref{fig:W2h-2}).

\begin{figure}[ht]
\centering\scalebox{0.7}{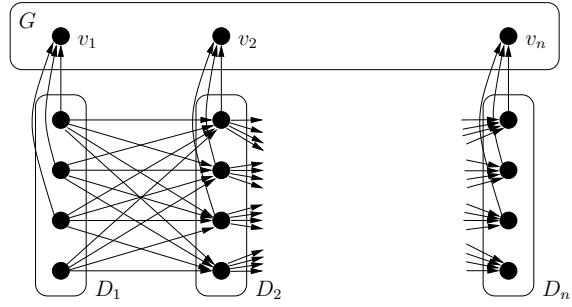}
\caption{Construction of $G'$ for $k=4$.
\label{fig:W2h-2}}
\end{figure}

\begin{itemize}
\item Construct a copy of $G$ and denote its vertices by $v_1,\ldots,v_n$.
\item For each $i\in\{1,\ldots,n\}$, construct a set of $k$ vertices $D_i$ and join $k-1$ vertices of this set with $v_i$
    by arcs.
\item For each $i\in \{2,\ldots,n\}$, join each vertex of $D_{i-1}$ with all vertices of $D_i$ by arcs.
\end{itemize}
Clearly, $G'$ is a directed acyclic graph. We let $b'=b+k$ and $p'=p+nk$. Let also $D=D_1\cup\ldots\cup D_n$. Notice that for
each $v\in V(G)$, $d_{G'}(v)=d_G(v)+k-1\leq k+2\leq\Delta$ as maximum degree of $G$ is 3. For $v\in D$, $d_{G'}(v)\leq
2k+1\leq\Delta$. Hence maximum degree of $G'$ is at most $\Delta$. We now claim that there is a set of at most $b$ vertices
$A\subseteq V(G)$ such that there exists an an induced subgraph $H$ of $G$ with at least $p$ vertices, $A\subseteq V(H)$ and
for any $v\in V(H')\setminus A$,  $d_{H}^-(v)\geq 1$ if and only if there is a set of at most $b'$ vertices $A'\subseteq
V(G')$ such that there exists an an induced subgraph $H'$ of $G'$ with at least $p'$ vertices, $A'\subseteq V(H')$ and for any
$v\in V(H)\setminus A$,  $d_{H'}^-(v)\geq k$.

Suppose that our original instance of \DAKC has a solution $(A,H)$. We let $A'=A\cup D_1$ and $H'=G'[V(H)\cup D]$. Then each
vertex $v\in D\setminus A'$ has $k$ in-neighbors in $D$. It remains to observe that each vertex $v$ of $G'$ from
$V(G)\setminus A'$ has at least one in-neighbor in $V(G)$ and $k-1$ in-neighbors in $D$. Therefore, $d_{G'}^-(v)\geq k$.

Assume now that $(A',H')$ is a solution for the constructed instance of \DAKC with $|A'|\leq b'$ and $|V(H)|\geq p'$. If
$|D\cap A'|<k$, then we claim that $D\cap V(H')\subseteq A'$. To prove it, suppose that $(V(H')\cap D)\setminus
A\neq\emptyset$ and consider the smallest index $i$ such that there is $v\in (V(H')\cap D_i)\setminus A$. Clearly, $i\geq 2$.
The vertex $v$ has in-neighbors only in $D_{i-1}$. By the choice of $i$, $D_{i-1}$ has at most $k-1$ vertices of $H'$, because
they can be only anchors and $|D\cap A'|<k$. Then $d_{H'}^-(v)<k$, a contradiction.

Then if $|D\cap A'|<k$, $V(H')\subseteq V(G)\cup A'$ and $|V(H')|\leq n+b+k\leq n+p+k<p'$ as $n\geq 3$ and $k\geq 2$. This
contradicts our assumption about size of $H'$. Hence, at least $k$ anchors are in $D$ and $|A'\setminus D|\leq b$. Let
$A=A'\setminus D$ and $H=H'-D$. If $v\in V(H)\setminus A$, then $d_{H'}^-(v)\geq k$ and $v$ has at most $k-1$ in-neighbors
from $D$ in $H'$. Then $v$ has at least one in-neighbor in $V(H)$ and $d_H^-(v)\geq 1$.
\end{proof}

The case of $k\geq \frac{\Delta}{2}$ the complexity of parameterization by $b$ on DAGs is left open. However we can show that
\DAKC is \classFPT on DAGs of maximum degree $\Delta$, when parameterized by $\Delta+p$.

\begin{theorem}\label{thm:dags}
For any positive integers $p$ and $\Delta$, \DAKC can be solved in time $2^{O(\Delta p)}\cdot n^2\log n$ for
$n$-vertex DAGs of maximum degree at most $\Delta$.
\end{theorem}

\begin{proof}
Consider an instance of \DAKC with an $n$-vertex directed acyclic graph $G$. Without loss of generality we can assume that
$b\leq p\leq n$.

We apply Lemma~\ref{lem:bounded} for $q=p$. In time $2^{O(\Delta p)}\cdot n\log n$ we either obtain a solution or conclude
that for any solution $(A,H)$, $H$ has size at least $p+1$. If we obtain a solution, we return it. Suppose that we got a
NO-answer. If $p=n$, then we return a NO-answer. Otherwise,  we select a sink $t\in V(G)$ using the fact that any directed
acyclic graph has at least one such vertex. Observe that we can assume that $t$ is not an anchor in any solution. Also if $t$
is included in a solution $H$ of size at least $p+1$, then $H-t$ is a solution of size at least $p$, because $t$ is not joined
by arcs with other vertices of $H$. Then we solve the instance $G-t$ of \DAKC recursively.

As each step is done in time $2^{O(\Delta p)}\cdot n\log n$ and the number of steps is at most $n$, the claim follows.
\end{proof}

Let us remark that this result can be easily extended for any class of directed acyclic graphs $\mathcal{G}$ such that the
corresponding class of underlaying graphs $\{G^*|G\in \mathcal{G}\}$ has (locally) bounded expansion by making use of  the
results by Dvorak et al.~\cite{DvorakKT10}.
Finally, what happens when the input graph is planar? We know that the problem is \classNP-complete on planar graphs for fixed
$k\geq 1$ and maximum degree $k+2$. Is the problem \classFPT\ on planar directed graphs when parameterized by the size of the
core $p$?

%\bibliographystyle{siam}
%\bibliography{Cores}

\end{document}